\documentclass{llncs}
\usepackage[utf8]{inputenc}
\usepackage[T1]{fontenc}
\usepackage{lmodern}
\usepackage{standalone}
\usepackage{textgreek}
\usepackage{microtype}
\usepackage[bookmarks,bookmarksopen,bookmarksdepth=2,colorlinks,unicode]{hyperref}
\usepackage{xspace}
\usepackage{float}
\usepackage{appendix}
\usepackage{pgfplots}\pgfplotsset{compat=1.18}
\usepackage{amsmath}
\usepackage{amsfonts}
\usepackage{amssymb}
\usepackage{braket}
\usepackage{mathtools}
\usepackage{siunitx}
\usepackage{enumitem}
\usepackage{algorithm}
\usepackage[noend]{algpseudocode}
\usepackage{array}
\usepackage{booktabs}
\usepackage{cleveref}
\usepackage{cite}
\usepackage{comment}
\usepackage{doi}
\usepackage{color}

\usetikzlibrary{cd}

\setlist[enumerate,1]{label={(\arabic*)}}

\makeatletter
\newcounter{algorithmicH}
\let\oldalgorithmic\algorithmic
\renewcommand{\algorithmic}{%
    \stepcounter{algorithmicH}
    \oldalgorithmic}
\renewcommand{\theHALG@line}{ALG@line.\thealgorithmicH.\arabic{ALG@line}}
\makeatother

\allowdisplaybreaks

\spnewtheorem{thm}[theorem]{Theorem}{\bfseries}{\itshape}
\spnewtheorem{defn}[theorem]{Definition}{\bfseries}{\itshape}
\spnewtheorem{prop}[theorem]{Proposition}{\bfseries}{\itshape}
\spnewtheorem{lem}[theorem]{Lemma}{\bfseries}{\itshape}
\spnewtheorem{cor}[theorem]{Corollary}{\bfseries}{\itshape}
\spnewtheorem{hyp}[theorem]{Hypothesis}{\bfseries}{\itshape}
\spnewtheorem{ex}[theorem]{Example}{\bfseries}{}
\spnewtheorem{rem}[theorem]{Remark}{\bfseries}{}

\crefname{thm}{theorem}{theorems}
\Crefname{thm}{Theorem}{Theorems}
\crefname{defn}{definition}{definitions}
\Crefname{defn}{Definition}{Definitions}
\crefname{prop}{proposition}{propositions}
\Crefname{prop}{Proposition}{Propositions}
\crefname{lem}{lemma}{lemmas}
\Crefname{lem}{Lemma}{Lemmas}
\crefname{cor}{corollary}{corollaries}
\Crefname{cor}{Corollary}{Corollaries}
\crefname{ex}{example}{examples}
\Crefname{ex}{Example}{Examples}
\crefname{rem}{remark}{remarks}
\Crefname{rem}{Remark}{Remarks}
\crefname{hyp}{hypothesis}{hypotheses}
\Crefname{hyp}{Hypothesis}{Hypotheses}

\makeatletter
\newcommand{\subalign}[1]{%
  \vcenter{%
    \Let@ \restore@math@cr \default@tag
    \baselineskip\fontdimen10 \scriptfont\tw@
    \advance\baselineskip\fontdimen12 \scriptfont\tw@
    \lineskip\thr@@\fontdimen8 \scriptfont\thr@@
    \lineskiplimit\lineskip
    \ialign{\hfil$\m@th\scriptstyle##$&$\m@th\scriptstyle{}##$\hfil\crcr
      #1\crcr
    }%
  }%
}
\makeatother

\newcommand{\F}{\mathbb{F}}                                         
\newcommand{\Fp}{\F_{p}}                                            
\newcommand{\Fq}{\F_{q}}                                            
\newcommand{\Fqn}{\Fq^{n}}                                          
\newcommand{\Fqnxn}{\Fq^{n \times n}}                               
\newcommand{\Fqm}{\Fq^{m}}                                          
\newcommand{\Fqen}{\F_{q^{n}}}                                      
\newcommand{\Fqx}{\Fq^\times}                                       
\DeclareMathOperator{\Orth}{Orth}                                   

\DeclareMathOperator{\id}{id}                                       
\DeclareMathOperator{\LP}{LP}                                       
\DeclareMathOperator{\Sym}{Sym}                                     
\DeclareMathOperator{\Tr}{Tr}                                       

\DeclareMathOperator{\CORR}{CORR}                                   
\DeclareMathOperator{\prob}{\mathbb{P}}                             

\newcommand*{\degree}[1]{\deg \left( #1 \right)}                    

\DeclareMathOperator{\wt}{wt}                                       

\newcommand{\Arion}{\textsf{Arion}\xspace}
\newcommand{\ArionHash}{\textsf{ArionHash}\xspace}

\newcommand{\Griffin}{\textsc{Griffin}\xspace}

\newcommand{\LowMC}{\texttt{LowMC}\xspace}

\newcommand{\Plonk}{\textsf{Plonk}\xspace}

\newcommand{\Poseidon}{\textsc{Poseidon}\xspace}

\newcommand{\RoneCS}{\textsf{R1CS}\xspace}

\newcommand{\ReinforcedConcrete}{\texttt{Reinforced Concrete}\xspace}

\begin{document}
    \title{Generalized Triangular Dynamical System: An Algebraic System for Constructing Cryptographic Permutations over Finite Fields}
    \author{Arnab Roy \and
            Matthias Johann Steiner}
    \institute{Alpen-Adria-Universit\"at Klagenfurt, Universit\"atsstraße 65-67, 9020 Klagenfurt am W\"orthersee, Austria  \\ \email{\{arnab.roy,matthias.steiner\}@aau.at}}

    \maketitle

    \begin{abstract}
        In recent years a new class of symmetric-key primitives over $\mathbb{F}_p$ that are essential to Multi-Party Computation and Zero-Knowledge Proofs based protocols have emerged.
        Towards improving the efficiency of such primitives, a number of new block ciphers and hash functions over $\mathbb{F}_p$ were proposed.
        These new primitives also showed that following alternative design strategies to the classical Substitution-Permutation Network (SPN) and Feistel Networks leads to more efficient cipher and hash function designs over $\mathbb{F}_p$ specifically for large odd primes $p$.

        In view of these efforts, in this work we build an \emph{algebraic framework} that allows the systematic exploration of viable and efficient design strategies for constructing symmetric-key (iterative) permutations over $\mathbb{F}_p$.
        We first identify iterative polynomial dynamical systems over finite fields as the central building block of almost all block cipher design strategies.
        We propose a generalized triangular polynomial dynamical system (GTDS), and based on the GTDS we provide a generic definition of an iterative (keyed) permutation over $\mathbb{F}_p^n$.

        Our GTDS-based generic definition is able to describe the three most well-known design strategies, namely SPNs, Feistel networks and Lai--Massey.
        Consequently, the block ciphers that are constructed following these design strategies can also be instantiated from our generic definition.
        Moreover, we find that the recently proposed \texttt{Griffin} design, which neither follows the Feistel nor the SPN design, can be described using the generic GTDS-based definition.
        We also show that a new generalized Lai--Massey construction can be instantiated from the GTDS-based definition.

        We further provide generic analysis of the GTDS including an upper bound on the differential uniformity and the correlation.
    \end{abstract}

    \section{Introduction}
    Constructing (keyed and unkeyed) permutations is at the center of designing some of the most broadly used cryptographic primitives like block ciphers and hash functions.
    After half a century of research, Feistel and Substitution-Permutation Networks (SPN) have emerged as the two dominant iterative design strategies for constructing unkeyed permutations or block ciphers.
    Another notable, although not much used design strategy is the Lai--Massey construction. Altogether, SPN, Feistel and Lai--Massey are at the core of some of the most well-known block ciphers such as AES \cite{AES-NIST,Daemen-AES}, DES \cite{DES77}, CLEFIA \cite{FSE:SSAMI07}, IDEA \cite{EC:LaiMas90}, etc.

    In the past few years, a new class of symmetric-key cryptographic functions (block ciphers, hash functions and stream ciphers) that are essential in privacy preserving
    cryptographic protocols based on Multi-Party Computation and Zero-Knowledge Proofs, have emerged.
    For efficiency reasons these primitives are designed over $\Fp$ (for large $p > 2$) as opposed to the classical symmetric primitives over $\F_{2^n}$ (typically for small $n$ e.g. $n \leq 8$). Following the classical approaches, a number of such symmetric-key functions were constructed either by utilizing the SPN or Feistel design principles.
    However, current research suggests that these traditional strategies are not the best choices for efficient primitives over $\Fp$.
    For example, the partial SPN-based hash function \Poseidon \cite{USENIX:GKRRS21} performs more efficiently in \RoneCS or \Plonk prover circuits than the generalized unbalanced Feistel-based construction \texttt{GMiMCHash} \cite{ESORICS:AGPRRRRS19}.
    Another recently proposed design - \Griffin \cite{Griffin}, follows neither SPN nor Feistel, and is more efficient in circuits than \texttt{GMiMCHash} and \Poseidon. In the literature these new primitives are often called Arithmetization Oriented (AO) primitives.

    An important and relevant question here is thus: {\em What is the space of possible design strategies for constructing (efficient) symmetric-key cryptographic permutations/functions over $\Fp$? And how to explore the possible design strategies systematically? }

    Moreover, given that such new cryptographic functions are inherently algebraic by design, their security is dictated by algebraic cryptanalytic techniques. For example, algebraic attacks (interpolation, Gr\"obner basis, GCD, etc.) \cite{AC:EGLORSW20,AC:ACGKLRS19,AC:AGRRT16,SAC:RoyAndSau20} are the main attack vectors in determining the security of \texttt{GMiMC}, \Poseidon, \texttt{MiMC} \cite{AC:AGRRT16}, etc.

    A well-defined generic algebraic design framework will prescribe a systematic approach towards exploring viable and efficient design strategies over $\Fp$.
    Such a generic framework will allow the design of new symmetric-key primitives and will shed new light into the algebraic properties of SPN- and Feistel-based designs, among others, over $\Fp$.
    A ``good'' generic framework should ultimately allow instantiation of primitives over $\Fq$ where $q = p^n$ for arbitrary primes $p$ and naturally encompass existing classical design strategies, such as SPN, Feistel and Lai--Massey.

    The primary aim of this work is to find such a general framework which describes iterative algebraic systems for constructing arithmetization oriented (keyed or unkeyed) permutations.

    \subsubsection{Study of generic frameworks and our work.} The study of generic frameworks for cryptographic constructions and their generic security analysis is a topic of high impact. It allows designers to validate their design strategies and gives recipes for possible design and analysis optimization advancements. Examples of research on generic design frameworks includes the studies on Even-Mansour (EM) design variants \cite{AC:Dutta20, C:CLLSS14, EC:CogSeu15, EC:DunKelSha12, FSE:FarPro15}, the generic security analysis of SPN \cite{C:LiuTesVai21}
    constructions, Sponge construction and variants thereof \cite{FSE:ADMA15, SAC:BDPV11, EC:GazTes16, C:FreGhoKom22}, etc. However, none of these works took an arithmetization oriented approach which might be due to the lack of practical applications for AO primitives in the past.

    The generic framework and its analysis of our work is based on the properties of polynomials over finite field $\Fq$. We can say that the EM construction or general SPN or Feistel constructions considered in previous works are much more generic in comparison to our proposed framework. For the cryptographic analysis
    in this work we only exploit the statistical (e.g. correlation, differential) and algebraic (polynomial degree) properties. This approach is comparable
    to the (statistical) security analysis \cite{C:LiuTesVai21} of the generic SPN.

    \subsection{Our Results}
    In this paper we lay out a generic strategy for constructing cryptographic (keyed and unkeyed) permutations combined with security analysis against differential cryptanalysis.

    We first discuss (\Cref{sec:blockcipher-permpoly}) that so-called orthogonal systems are the only polynomial systems suitable to represent (keyed) permutations and henceforth block ciphers over finite fields.

    We then propose a novel algebraic system (in \Cref{sec:gtds_framework}) that is the foundation for constructing generic iterative permutations.
    More specifically, we construct a polynomial dynamical system over a finite field $\Fq$ (where $q = p^n$ with $p$ a prime and $n \geq 1$) that we call Generalized Triangular Dynamical System (GTDS).
    We then provide a generic definition of iterative (keyed) permutations using the GTDS and a linear/affine permutation.
    We show (in \Cref{Sec: examples of block ciphers}) that our GTDS-based definition of iterative permutations is able to describe the SPN, different types of Feistel networks and the Lai--Massey construction. Consequently, different block ciphers that are instantiations of these design strategies can also be instantiated from the GTDS-based permutation.

    Beyond encompassing these well-known design strategies, our framework provides a systematic way to study different algebraic design strategies and security of permutations (with or without key). This is extremely useful in connection with the recent design efforts for constructing block ciphers and hash function over $\Fp$ where $p$ is a large prime. For example, GTDS already covers the recently proposed partial SPN design strategy \cite{EC:GLRRS20} used in designing block ciphers and hash functions \cite{USENIX:GKRRS21}.

    Our GTDS-based definition of iterative permutations allows for instantiations of new (keyed) permutations.
    For example, the recently proposed construction \Griffin can also be instantiated from our generic definition of an iterative permutation.
    Moreover, using our generic definition we propose a generalization (\Cref{Sec: Lai--Massey}) of the Lai--Massey design strategy.
    A new efficient and secure cryptographic permutation (and hash function) \cite{Arion} with low multiplicative complexity is also instantiated from our generic definition.

    In \Cref{Sec: analysis of GTDS} we perform a generic analysis to bound the differential uniformity as well as the correlation of the GTDS.

    Our generic constructions, definitions and results holds for arbitrary $p$. However, our main aim is to propose an algebraic
    framework for constructing primitives and provide generic (security) analysis over $\Fp$ for (large) $p > 2$. The security analysis given in this paper can be refined and improved for $p = 2$. Our (security) analysis is not aimed for binary extension field and should be viewed as
    generic analysis for $p > 2$. However, the GTDS-based construction(s) proposed in this paper can be applied over $\Fq$ (where $q = p^n$ with $p$ a prime and $n \geq 1$).

    \section{Block Ciphers and Permutation Polynomials} \label{sec:blockcipher-permpoly}
    In general a deterministic (block) cipher can be described as a pair of keyed mappings
    \begin{equation}
        F: \mathcal{M} \times \mathcal{K} \to \mathcal{C}, \qquad F^{-1}: \mathcal{C} \times \mathcal{K} \to \mathcal{M},
    \end{equation}
    where $\mathcal{M}$, $\mathcal{K}$ and $\mathcal{C}$ denote the message, key and cipher domain and such that $F^{-1}(\_, \mathbf{k}) \circ F(\_, \mathbf{k}) = \id_\mathcal{M}$ for all $\mathbf{k} \in \mathcal{K}$.
    In practice the domains $\mathcal{M}$, $\mathcal{K}$ and $\mathcal{C}$ are finite, thus by \cite[Theorem~72]{Bard-AlgebraicCryptanalysis} any cipher can be modeled as a mapping between vector spaces over finite fields.
    In this work we will assume that $\mathcal{M} = \mathcal{C} = \Fqn$ and $\mathcal{K} = \Fq^{n \times r}$, where $r, n \geq 1$, $q$ is a prime power and $\Fq$ is the field with $q$ elements.
    For a block cipher we also require that $F$ is a keyed permutation over $\Fqn$, i.e., for all $\mathbf{k} \in \Fq^{n \times r}$ the function $F(\_, \mathbf{k})$ is a permutation.
    Note that for any function $F: \Fqn \to \Fq$ we can find via interpolation a unique polynomial $P \in \Fq[x_1, \dots, x_n]$ with degree less than $q$ in each variable such that $F(\mathbf{x}) = P(\mathbf{x})$ for all $\mathbf{x} \in \Fqn$.
    Therefore, we will also interpret all ciphers as vectors of polynomial valued functions.
    We recall the formal (algebraic) notion of polynomial vectors that induce a permutation.
    \begin{defn}[See {\cite[7.34., 7.35. Definition]{Niederreiter-FiniteFields}}]
        Let $\Fq$ be a finite field.
        \begin{enumerate}
            \item A polynomial $f \in \Fq [x_1, \allowbreak \dots, x_n]$ is called a permutation polynomial if the equation $f (x_1, \allowbreak \dots, x_n) = \alpha$ has $q^{n - 1}$ solutions in $\Fqn$ for each $\alpha \in \Fq$.

            \item A system of polynomials $f_1, \dots, \allowbreak f_m \in \Fq[x_1, \dots, x_n]$, where $1 \leq m \leq n$,  is said to be orthogonal if the system of equations $f_1( x_1, \dots, x_n) = \alpha_1, \dots, \allowbreak f_m(x_1, \dots, x_n) = \alpha_m$ has exactly $q^{n - m}$ solutions in $\Fqn$ for each $(\alpha_1, \dots, \allowbreak \alpha_m) \in \F_{q}^{m}$.
        \end{enumerate}
    \end{defn}

    Over $\F_2$, the permutation polynomials are known as balanced functions \cite{FSE:BouDeCDeC11} in the cryptography/computer science literature.

    \begin{rem}
        It is immediate from the definition that every subset of an orthogonal system is also an orthogonal system.
        In particular, every polynomial in an orthogonal system is also a multivariate permutation polynomial.
        If for an orthogonal system $m = n$, then the orthogonal system induces a permutation on $\Fqn$.
        Moreover, if we restrict orthogonal systems to the $\Fq$-algebra of polynomial valued functions $\Fq [x_1, \dots, x_n] / (x_1^q - x_1, \dots, x_n^q - x_n)$, that is the polynomials with degree less than $q$ in each variable, then the orthogonal systems of size $n$ form a group under composition.
        If we denote this group with $\Orth_n \left( \Fq \right)$, then one can establish the following isomorphisms of groups $\Orth_n \left( \Fq \right) \cong \Sym \left( \Fqn \right) \cong \Sym \left( \Fqen \right)$, where $\Sym \left( \_ \right)$ denotes the symmetric group (cf.\ \cite[7.45. Corollary]{Niederreiter-FiniteFields}).
    \end{rem}

    Since one of our main interests is in keyed permutations let us extend the definition of orthogonal systems.
    In general, we will denote with $x$ the plaintext variables and with $y$ the key variables.
    \begin{defn}\label{Def: keyed permutations}
        Let $\Fq$ be a finite field.
        \begin{enumerate}
            \item Let $F: \Fq^{n_1} \times  \Fq^{n_2} \to \Fq^{n_1}$ be a function. We call $F$ a keyed permutation, if for any fixed $\mathbf{y} \in \Fq^{n_2}$ the function $F(\_, \mathbf{y}): \Fq^{n_1} \to \Fq^{n_1}$ induces a permutation.

            \item Let $f_1, \dots, f_m \in \Fqn[x_1, \dots, x_{n_1}, y_1, \dots, y_{n_2}]$, where $1 \leq m \leq n_1$ be polynomials. We call $f_1, \dots, f_m$ a keyed orthogonal system, if for any fixed $(y_1, \dots, y_{n_2}) \in \Fq^{n_2}$ the system $f_1, \dots, f_m$ is an orthogonal system.
        \end{enumerate}
    \end{defn}
    \begin{rem}
        \begin{enumerate}
            \item Note that in our definition we allow for trivial keyed permutations, i.e., permutations that are constant in the key variable.
            In particular, every permutation $F: \Fqn \to \Fqn$ induces a keyed permutation $\hat{F}: \Fqn \times \Fqm \to \Fqm$ via $\hat{F} (\mathbf{x}, \mathbf{y}) = F (\mathbf{x})$ for any $m \in \mathbb{Z}_{\geq 1}$.

            \item A keyed orthogonal system is also an orthogonal system in $\Fq[x_1, \dots, x_{n_1}, y_1, \allowbreak \dots, y_{n_2}]$.
            Suppose we are given a keyed orthogonal system $f_1, \dots, f_m \in \Fq[x_1, \dots, x_{n_1}, y_1, \dots, \allowbreak y_{n_2}]$ and equations $f_i(\mathbf{x}, \mathbf{y}) = \alpha_i$, where $\alpha_i \in \Fq$.
            If we fix $\mathbf{y}$ then we have $q^{n_1 - m}$ many solutions for $\mathbf{x}$.
            There are $q^{n_2}$ possible choices for $\mathbf{y}$, so the system has $q^{n_1 + n_2 - m}$ solutions.
            Hence, our definition of keyed orthogonal systems does not induce any essentially new structure, it is merely semantic.
        \end{enumerate}
    \end{rem}
    As intuition suggests keyed orthogonal systems are well-behaved under iteration. We state the following theorem for completeness.
    \begin{thm}\label{Th: keyed orthogonal systems under iteration}
        Let $\Fq$ be a finite field.
        The keyed polynomial system $f_1, \dots, f_m \in \Fq[x_1, \dots, x_{n_1}, \allowbreak y_1, \allowbreak \dots, y_{n_2}]$ is keyed orthogonal if and only if the system $g_1 (f_1, \dots, \allowbreak f_m, \allowbreak y_1, \dots, y_{n_2}), \dots, \allowbreak g_m (f_1, \allowbreak \dots, f_m, y_1, \dots, y_{n_2}) \in \Fq[x_1, \dots, x_{n_1}, y_1, \allowbreak \dots, y_{n_2}]$ is keyed orthogonal for every keyed orthogonal system $g_1, \dots, g_m \in \Fq[x_1, \dots, x_m, \allowbreak y_1, \dots, \allowbreak y_{n_2}]$.
    \end{thm}
    \begin{proof}
        ``$\Leftarrow$'': If we choose $g_i = x_i$, then by assumption the equations
        \begin{equation*}
            g_i (f_1, \dots, f_m, y_1, \dots, y_{n_2}) = f_i (x_1, \dots, x_{n_1}, y_1, \dots, y_{n_2}) = \beta_i,
        \end{equation*}
        where $1 \leq i \leq m$, have $q^{n_1 - m}$ many solutions for every fixed $(y_1, \dots, y_{n_2}) \in \Fq^{n_2}$.
        I.e., $f_1, \dots, f_m$ is a keyed orthogonal system.
        \newline
        ``$\Rightarrow$'': Suppose we are given a system of equations
        \begin{equation*}
            \begin{split}
                g_1(f_1, \dots, f_m, y_1, \dots, y_{n_2}) &= \beta_1, \\
                \dots \\
                g_m(f_1, \dots, f_m, y_1, \dots, y_{n_2})  &= \beta_m,
            \end{split}
        \end{equation*}
        where $\beta_1, \dots, \beta_m \in \Fq$ and $\{ f_i \}_{1 \leq i \leq n_1}$ and $\{ g_i \}_{1 \leq i \leq m}$ are keyed orthogonal systems.
        Fix $\mathbf{y} = (y_1, \dots, y_{n_2}) \in \Fq^{n_2}$ and substitute $\hat{x}_i = f_i$, then the equations $g_i (\hat{x}_1, \dots, \hat{x}_m, \mathbf{y}) = \beta_i$ have a unique solution for the $\hat{x}_i$'s.
        Since $\mathbf{y}$ is fixed also the equations $\hat{x}_i = f_i$ admit $q^{n_2 - m}$ many solutions.
        Therefore, the composition of keyed orthogonal systems is again keyed orthogonal. \qed
    \end{proof}

    In practice keyed orthogonal systems are usually derived from orthogonal systems by a simple addition of the key variables before or after an evaluation of a function.
    \begin{ex}
        If $F: \Fqn \to \Fqn$ is a permutation, then
        \begin{equation*}
            F (\mathbf{x} + \mathbf{y}) \qquad\text{and} \qquad F (\mathbf{x}) + \mathbf{y}
        \end{equation*}
        are keyed permutations.
    \end{ex}

    \section{Generalized Triangular Dynamical Systems}\label{sec:gtds_framework}
    We propose the generalized triangular dynamical system (GTDS) as the main ingredient when designing a block cipher. The GTDS is also the main ingredient in unifying different design principles proposed in the literature such as SPN and Feistel networks.
    \begin{defn}[Generalized triangular dynamical system]\label{Def: generlized triangular dynamical system}
        Let $\Fq$ be a finite field, and let $n \geq 1$.
        For $1 \leq i \leq n$, let $p_i \in \Fq[x]$ be permutation polynomials, and for $1 \leq i \leq n - 1$, let $g_i, h_i \in \Fq[x_{i + 1}, \dots, x_n]$ be polynomials such that the polynomials $g_i$ do not have zeros over $\Fq$.
        Then we define a generalized triangular dynamical system $\mathcal{F} = \{ f_1, \dots, f_n \}$ as follows
        \begin{equation*}\label{Equ: generalized dynamical system}
            \begin{split}
                f_1(x_1,\dots,x_n)		&= p_1(x_1) \cdot g_1(x_2,\dots,x_n) + h_1(x_2,\dots,x_n),  \\
                f_2(x_1,\dots,x_n)		&= p_2(x_2) \cdot g_2(x_3,\dots,x_n) + h_2(x_3,\dots,x_n),	\\
                & \dots											                                    \\
                f_{n-1}(x_1,\dots,x_n)	&= p_{n-1}(x_{n-1}) \cdot g_{n-1}(x_n) + h_{n-1}(x_n),	    \\
                f_n(x_1,\dots,x_n)		&= p_n(x_n).
            \end{split}
        \end{equation*}
    \end{defn}
    Note that a GTDS $\mathcal{F} = \{ f_1, \dots, f_n \}$ must be considered as ordered tuple of polynomials since in general the order of the $f_i$'s cannot be interchanged.
    \begin{prop}
        A generalized triangular dynamical system is an orthogonal system.
    \end{prop}
    \begin{proof}
        Suppose for $1 \leq i \leq n$ we are given equations
        \begin{equation*}
            f_i(x_i, \dots, x_n) = \alpha_i,
        \end{equation*}
        where $\alpha_i \in \Fq$. To solve the system we work upwards.
        The last polynomial $f_n$ is a univariate permutation polynomial, so we can find a unique solution $\beta_n$ for $x_n$. We plug this solution into the next equation, i.e.,
        \begin{equation*}
            f_{n - 1} (x_{n - 1}, \beta_n) = p_{n - 1} (x_{n - 1}) \cdot g_{n - 1} (\beta_n) + h_{n - 1} (\beta_n).
        \end{equation*}
        To solve for $x_{n - 1}$ we subtract $h_{n - 1} (\beta_n)$, divide by $g_{n - 1} (\beta_n)$, this division is possible since $g_i (x_{i + 1}, \dots, x_n) \neq 0$ for all $(x_{i + 1}, \dots, x_n) \in \Fq^{n - i}$, and invert $p_{n - 1}$. Iterating this procedure we can find a unique solution for all $x_i$. \qed
    \end{proof}
    \begin{cor}
        The inverse orthogonal system $\mathcal{F}^{-1} = \{ \tilde{f}_1, \dots, \tilde{f}_n \}$ to the generalized triangular dynamical system $\mathcal{F} = \{ f_1, \dots, f_n \}$ is given by
        \begin{equation*}
            \begin{split}
                \tilde{f}_1 (x_1, \dots, x_n) &= p_1^{-1} \left( \Big( x_1 - h_1 \big( \tilde{f}_2, \dots, \tilde{f}_n \big) \Big) \cdot \Big( g_1 \big( \tilde{f}_2, \dots, \tilde{f}_n \big) \Big)^{q - 2} \right) \\
                \tilde{f}_2 (x_1, \dots, x_n) &= p_2^{-1} \left( \Big( x_2 - h_2 \big( \tilde{f}_3, \dots, \tilde{f}_n \big) \Big) \cdot \Big( g_2 \big( \tilde{f}_3, \dots, \tilde{f}_n \big) \Big)^{q - 2} \right) \\
                & \dots \\
                \tilde{f}_{n - 1} (x_1, \dots, x_n) &= p_{n - 1}^{-1} \left( \Big( x_{n - 1} - h_{n - 1} \big( \tilde{f}_n \big) \Big) \cdot \Big( g_{n - 1} \big( \tilde{f}_n \big) \Big)^{q - 2} \right) \\
                \tilde{f}_n (x_1, \dots, x_n) &= p_n^{-1} (x_n).
            \end{split}
        \end{equation*}
    \end{cor}
    \begin{proof}
        If we consider $\mathcal{F}$ and $\mathcal{F}^{-1}$ in $\Fq[x_1, \dots, x_n] / \big( x_1^q - x_1, \dots, \allowbreak x_n^q - x_n \big)$, then it is easy to see that $\mathcal{F}^{-1} \circ \mathcal{F} = \mathcal{F} \circ \mathcal{F}^{-1} = \id$. \qed
    \end{proof}

    Note that the Triangular Dynamical System introduced by Ostafe and Shparlinski \cite{Ostafe-DegreeGrowth} is a special case of our GTDS.
    In particular, if we choose $p_i(x_i) = x_i$ for all $i$ and impose the condition that each polynomial $g_i$ has a unique leading monomial, i.e.,
    \begin{equation}
        g_i (x_{i + 1}, \dots, x_n) = x_{i + 1}^{s_{i, i + 1}} \cdots x_n^{s_{i, n}} + \tilde{g}_i (x_{i + 1}, \dots, x_n),
    \end{equation}
    where
    \begin{align}
        \degree{\tilde{g}} &< s_{i, i + 1} + \ldots + s_{i, n}, \text{ and}\\
        \degree{h_i} &\leq \degree{g_i}
    \end{align}
    for $i = 1, \dots, n - 1$, then we obtain the original triangular dynamical systems.
    Notice that under iteration these systems exhibit a property highly uncommon for general polynomial dynamical systems: polynomial degree growth (see \cite[\S2.2]{Ostafe-DegreeGrowth}).

    \subsection{GTDS and (Keyed) Permutations}
    In practice, every keyed permutation or block cipher (in cryptography) is constructed using an iterative structure where round functions are iterated a fixed number of times. Using the GTDS we first define such a round function.
    In this section $n \in \mathbb{N}$ denotes the number of field elements constituting a block and $r \in \mathbb{N}$ denotes the number of rounds of
    an iterative permutation.
    \begin{defn}[Round function]
        Let $\Fq$ be a finite field, let $n \geq 1$ be an integer, let $\mathbf{A} \in \Fqnxn$ be an invertible matrix, and let $\mathbf{b}\in \Fqn$ be a vector.
        Then, the affine mixing layer is described by the map
        \begin{equation*}
            \mathcal{L}: \Fqn \to \Fqn, \qquad \mathbf{x} \mapsto \mathbf{A} \cdot \mathbf{x} + \mathbf{b},
        \end{equation*}
        and the key addition is described by the map
        \begin{equation*}
            \mathcal{K}: \Fqn \times \Fqn \to \Fqn, \qquad \left( \mathbf{x}, \mathbf{k} \right) \mapsto \mathbf{x} + \mathbf{k}.
        \end{equation*}
        We abbreviate $\mathcal{K}_\mathbf{k} = \mathcal{K} (\_, \mathbf{k})$.
        Let $\mathcal{F} \subset \Fq[x_1, \dots, x_n]$ be a GTDS or a composition of two or more GTDS and affine permutations.
        Then the round function of a block cipher is defined as the following composition
        \begin{equation*}
            \mathcal{R}: \Fqn \times \Fqn \to \Fqn, \qquad \left( \mathbf{x}, \mathbf{k} \right) \mapsto \mathcal{K}_\mathbf{k} \circ \mathcal{L} \circ \mathcal{F} \left( \mathbf{x} \right).
        \end{equation*}
        We also abbreviate $\mathcal{R}_\mathbf{k} = \mathcal{R} (\_, \mathbf{k})$.
    \end{defn}
    It is obvious that $\mathcal{R}$ is a keyed permutation, hence it also is a keyed orthogonal system of polynomials in the sense of \Cref{Def: keyed permutations}.
    Now we can introduce our generalized notion of block ciphers which encompasses almost all existing block ciphers.
    \begin{defn}[An algebraic description of keyed permutations]\label{Def: block cipher}
        Let $\Fq$ be a finite field, let $n, r \geq 1$ be integers, and let $\mathbf{K} \in \Fq^{n \times (r + 1)}$ be a matrix.
        We index the columns of $\mathbf{K}$ by $0, \dots, r$, the $i$\textsuperscript{th} column $\mathbf{k}_i$ denotes the $i$\textsuperscript{th} round key.
        Let $\mathcal{K}: \Fqn \times \Fqn \to \Fqn$ be the key addition function, and let $\mathcal{R}^{(1)}, \dots, \mathcal{R}^{(r)}: \Fqn \times \Fqn \to \Fqn$ be the round functions.
        Then a block cipher is defined as the following composition
        \begin{equation*}
            \mathcal{C}_{r}: \Fqn \times \Fq^{n \times (r + 1)} \to \Fqn,  \qquad \left( \mathbf{x}, \mathbf{K} \right) \mapsto \mathcal{R}_{\mathbf{k}_r}^{(r)} \circ \cdots \circ \mathcal{R}_{\mathbf{k}_1}^{(1)} \circ \mathcal{K}_{\mathbf{k}_0} \left( \mathbf{x} \right).
        \end{equation*}
        We abbreviate $\mathcal{C}_{r, \mathbf{K}} = \mathcal{C}_{r} (\_, \mathbf{K})$, and if the round functions are clear from context or identical, then we also abbreviate $\mathcal{R}_\mathbf{k}^r = \mathcal{R}_{\mathbf{k}_r}^{(r)} \circ \cdots \circ \mathcal{R}_{\mathbf{k}_1}^{(1)}$.
    \end{defn}
    For the remaining parts of the paper a keyed permutation or a block cipher should be understood as a function described as in \Cref{Def: block cipher}, unless specified otherwise.
    We stress that a generic definition of an iterative block cipher may only use the notion of round key(s) (as defined with $\mathbf{K}$ in \Cref{Def: block cipher}) and does not require explicit definition of a key scheduling function.
    The specific definition of a key scheduling function can depend on the input key size and specific instantiations of the iterative block cipher.
    Also, for most of the cryptographic literature the generic definition, (security) analysis and security proofs of iterative block ciphers (e.g. SPN, Even-Mansour etc.) only use the notion of round keys \cite{AC:LamPatSeu12,AC:DDKS14,EC:CogSeu15}, not an explicit scheduling function.

    \section{Instantiating Block Ciphers}\label{Sec: examples of block ciphers}
    In this section we will show that the GTDS-based algebraic definition of iterative permutations is able to describe different design strategies.

    We note with respect to GTDS that well-known design strategies such as SPN, partial SPN, Feistel, generalized Feistel and Lai--Massey are constructed with trivial polynomials $g_i$ in the GTDS, namely $g_i = 1$.

    \subsection{Feistel Networks}\label{Sec: Feistel network}
    For simplicity, we only show how the GTDS based algebraic definition can describe the unbalanced Feistel with expanding round function.
    The classical two branch Feistel is then a special case of the unbalanced expanding one.
    Moreover, it is straight-forward to show that GTDS-based algebraic definition can describe other types of Feistel networks such as unbalanced Feistel with expanding round functions, Nyberg's GFN, etc.

    \subsubsection{Unbalanced Feistel.}
    Let $n > 1$, and let $f \in \Fq[x]$ be any function represented by a polynomial.
    The unbalanced Feistel network with expanding round function is defined as
    \begin{equation}
        \begin{pmatrix}
            x_1 \\
            \vdots \\
            x_n
        \end{pmatrix}
        \mapsto
        \begin{pmatrix}
            x_n \\
            x_1 + f \left( x_n \right) \\
            \vdots \\
            x_{n - 1} + f \left( x_n \right)
        \end{pmatrix}
        .
    \end{equation}
    The GTDS
    \begin{equation}
        \begin{split}
            f_i (x_1, \dots, x_n) &= x_i + f(x_n),\quad 1 \leq i \leq n - 1, \\
            f_n (x_1, \dots, x_n) &= x_n,
        \end{split}
    \end{equation}
    together with the shift permutation
    \begin{equation}
        (x_1, \dots, x_{n - 1}, x_n) \mapsto (x_n, x_1, \dots, x_{n - 1})
    \end{equation}
    describe the unbalanced Feistel network with expanding round function.

    \subsection{Substitution-Permutation Networks}\label{Sec: SPN}
    In \cite[\S7.2.1]{Katz-Cryptography} a handy description of Substitution-Permutation networks (SPN) was given. Let $S \in \Fq[x]$ be a permutation polynomial, the so called S-box. Then the round function of a SPN consists of three parts:
    \begin{enumerate}[label=(\arabic*)]
        \item Addition of the round keys.

        \item Application of the S-box, i.e.,
        \begin{equation*}
            (x_1, \dots, x_n) \mapsto \big( S(x_1), \dots, S(x_n) \big).
        \end{equation*}

        \item Permutation and mixing of the blocks.
    \end{enumerate}
    The mixing in the last step is usually done via linear/affine transformations.
    In this case the GTDS of a SPN reduces to
    \begin{equation}\label{Equ: S-box GTDS}
        f_i (x_1, \dots, x_n) = S(x_i),
    \end{equation}
    where $1 \leq i \leq n$.
    If the last step is not linear then one either must introduce additional GTDS as round functions or modify the GTDS in \Cref{Equ: S-box GTDS}.

    \subsubsection{AES-128.}
    At the time of writing the most famous SPN is the AES family \cite{AES,Daemen-AES}.
    If we use the description of AES-128 given in \cite{FSE:BucPysWei06}, then it is easy to see that AES-128 is also covered by our definition of block ciphers.
    AES-128 is defined over the field $\F = \F_{2^8}$ and has $16$ blocks, i.e., it is a keyed permutation over $\F^{16}$. The AES-128 S-box is given by
    \begin{alignat}{2}
        S: \F &\to \F, \nonumber \\
        x &\mapsto \texttt{05} x^{254} &&+ \texttt{09} x^{253} + \texttt{F9} x^{251} + \texttt{25} x^{247} + \texttt{F4} x^{239} \\ &\phantom{x} &&+ x^{223} + \texttt{B5} x^{191} + \texttt{8F} x^{127} + \texttt{63}, \nonumber
    \end{alignat}
    and the GTDS of AES-128 is given by \Cref{Equ: S-box GTDS}.

    Let us now describe the permuting and mixing of the blocks via linear transformations.
    The \texttt{ShiftRows} operations can be described with the block matrix
    \begin{equation}
        D_\text{SR} =
        \begin{pmatrix}
            D_{\text{SR}_0} &   0   &   0   &   0   \\
            0   &   D_{\text{SR}_1} &   0   &   0   \\
            0   &   0   &   D_{\text{SR}_2} &   0   \\
            0   &   0   &   0   &   D_{\text{SR}_3}
        \end{pmatrix}
        \in \F^{16 \times 16},
    \end{equation}
    where
    \begin{equation}
        D_{\text{SR}_t} = \left( \Delta_{i, (j - t) \mod 4} \right) \in \F^{4 \times 4},
    \end{equation}
    and $\Delta_{i, j}$ is the Kronecker delta.
    The \texttt{MixColumns} operation can be described as the following tensor product
    \begin{equation}
        D_\text{MC} =
        \begin{pmatrix}
            \texttt{02}   &   \texttt{03}   &   \texttt{01}   &   \texttt{01}   \\
            \texttt{01}   &   \texttt{02}   &   \texttt{03}   &   \texttt{01}   \\
            \texttt{01}   &   \texttt{01}   &   \texttt{02}   &   \texttt{03}   \\
            \texttt{03}   &   \texttt{01}   &   \texttt{01}   &   \texttt{02}   \\
        \end{pmatrix}
        \otimes I_4 \in \F^{16 \times 16},
    \end{equation}
    where the entries in the left matrix are hexadecimal representations of field elements.
    The linear mixing layer $\mathcal{L}$ of AES-128 can now be represented by the following matrix $D$
    \begin{equation}
        D = P \cdot D_\text{MC} \cdot D_\text{SR} \cdot P,
    \end{equation}
    where $P \in \F^{16 \times 16}$ denotes the transposition matrix.
    In the last round the \texttt{MixColumns} operation is dropped, hence $\tilde{\mathcal{L}}$ is represented by $\tilde{D}$
    \begin{equation}
        \tilde{D} = P \cdot D_\text{SR} \cdot P.
    \end{equation}
    Similarly, we can also describe the key schedule of AES-128.

    \subsubsection{Partial SPN.}
    In a partial SPN the S-box is only applied to some input variables and not all of them.
    This construction was proposed for ciphers like \LowMC \cite{EC:ARSTZ15}, the Hades design strategy \cite{EC:GLRRS20} and the \Poseidon family \cite{USENIX:GKRRS21} that are efficient in the MPC setting.
    Clearly, any partial SPN is also covered by the GTDS.

    \subsection{Lai--Massey Ciphers and GTDS}\label{Sec: Lai--Massey}
    Another well-known design strategy for block ciphers is the Lai--Massey design which was first introduced in \cite{Lai-Massey}.
    For two branches let $g \in \Fq[x]$ be a polynomial, then the round function of the Lai--Massey cipher is defined as
    \begin{equation}
        \mathcal{F}_\text{LM}:
        \begin{pmatrix}
            x \\ y
        \end{pmatrix}
        \mapsto
        \begin{pmatrix}
            x + g (x - y) \\ y +  g(x - y)
        \end{pmatrix}.
    \end{equation}
    Since the difference between the branches is invariant under application of $\mathcal{F}_\text{LM}$ it is possible to invert the construction.
    At the first look it may appear that the Lai--Massey can not be described with GTDS. However, a careful analysis shows one round of Lai--Massey is in fact a composition of a Feistel Network and two linear permutations.
    We consider the following triangular dynamical systems
    \begin{equation}
        \mathcal{F}_1 (x, y) =
        \begin{pmatrix}
            x - y \\ y
        \end{pmatrix}
        ,\quad
        \mathcal{F}_2 (x, y) =
        \begin{pmatrix}
            x \\ y + g(x)
        \end{pmatrix}
        ,\quad
        \mathcal{F}_3 (x, y) =
        \begin{pmatrix}
            x + y \\ y
        \end{pmatrix}.
    \end{equation}
    Then, it is easily checked that $\mathcal{F}_\text{LM} = \mathcal{F}_3 \circ \mathcal{F}_2 \circ \mathcal{F}_1$.

    \subsubsection{Generalized Lai--Massey.}
    Recently, a generalization of the Lai--Massey was proposed in \cite[\S3.3]{ToSC:GOPS22} by Grassi et al.
    It is based on the following observation: If one is given field elements $\omega_1, \dots, \omega_n \in \Fq$ such that $\sum_{i = 1}^{n} \omega_i = 0$, then the mapping
    \begin{equation}
        \begin{pmatrix}
            x_1 \\ \vdots \\ x_n
        \end{pmatrix}
        \mapsto
        \begin{pmatrix}
            x_1 + g (\sum_{i = 1}^{n} \omega_i x_i) \\ \vdots \\ x_n + g (\sum_{i = 1}^{n} \omega_i x_i)
        \end{pmatrix}
    \end{equation}
    is invertible for any polynomial $g \in \Fq [x]$.

    We will use this observation to propose an even more general version of the Lai--Massey from the GTDS and linear permutations.
    \begin{defn}[Generalized Lai--Massey]\label{Def: generlized dynamical system of Lai Massey cipher}
        Let $\Fq$ be a finite field, and let $n \geq 2$ be an integer.
        Let $\omega_1, \dots, \omega_n \in\Fq$ be such that $\sum_{i = 1}^n \omega_i = 0$, and denote with $m$ the largest index $1 \leq i \leq n$ such that $\omega_i$ is non-zero.
        For $1 \leq i \leq n$ let $p_i \in \Fq[x]$ be permutation polynomials, and let $g\in \Fq[x, x_{m + 1}, \dots, x_n]$ be a polynomial.
        Then we define the generalized Lai--Massey $\mathcal{F}_\text{LM} = \{ f_1, \dots, f_n \}$ as follows
        \begin{equation*}
            \begin{split}
                f_1 (x_1, \dots, x_n) &= p_1 (x_1) + g \left( \sum_{i = 1}^{m} \omega_i \cdot p_i (x_i), x_{m + 1}, \dots, x_n \right), \\
                &\dots \\
                f_m (x_1, \dots, x_n) &= p_m (x_m) + g \left( \sum_{i = 1}^{m} \omega_i \cdot p_i (x_i), x_{m + 1}, \dots, x_n \right), \\
                f_{m + 1} (x_1, \dots, x_n) &= p_{m + 1} (x_{m + 1}), \\
                &\dots \\
                f_n (x_1, \dots, x_n) &= p_n (x_n).
            \end{split}
        \end{equation*}
    \end{defn}
    \begin{rem}
        If $n \equiv 0 \mod 2$, then it is evident from the first equation in the proof of \cite[Proposition~5]{EPRINT:GOPS21} that Grassi et al.'s generalized Lai--Massey permutation is also covered by \Cref{Def: generlized dynamical system of Lai Massey cipher} and a linear transformation.
    \end{rem}
    For completeness, we establish that the generalized Lai--Massey is indeed invertible.
    \begin{lem}
        Let $\Fq$ be a finite field.
        The generalized Lai--Massey is an orthogonal system.
    \end{lem}
    \begin{proof}
        Suppose we are given equations $f_i (x_1, \dots, x_n) = \alpha_i$, where $\alpha_i \in \Fq$.
        For $i = m + 1, \dots, n$ we simply invert $p_i$ to solve for $x_i$.
        For $i = 1, \dots, m$ we compute $\sum_{i = 1}^{m} \omega_i f_i = \sum_{i = 1}^{m} \omega_i p_i (x_i) = \sum_{i = 1}^{m} \omega_i \alpha_i$ = $\alpha$.
        Now we plug $\alpha$ and the solutions for $x_{m + 1}, \dots, x_n$ into the polynomial $g$ in the first $m$ equations, rearrange them, and invert the univariate permutation polynomials to obtain a unique solution. \qed
    \end{proof}

    Before we prove the reduction of the generalized Lai--Massey to the GTDS we explain the rationale behind \Cref{Def: generlized dynamical system of Lai Massey cipher}.
    Usually, in the Lai--Massey the polynomial $g$ is added to all the branches, but our definition allows the concatenation of two independent Lai--Massey permutations
    \begin{equation}
        \begin{pmatrix}
            x_1 \\ x_2 \\ x_3 \\ x_4
        \end{pmatrix}
        \mapsto
        \begin{pmatrix}
            x_1 + g_1 (x_1 - x_2) \\
            x_2 + g_1 (x_1 - x_2) \\
            x_3 + g_2 (x_3 - x_4) \\
            x_4 + g_2 (x_3 - x_4)
        \end{pmatrix}
        ,
    \end{equation}
    or the construction of intertwined Lai--Massey permutations
    \begin{equation}
        \begin{pmatrix}
            x_1 \\ x_2 \\ x_3 \\ x_4
        \end{pmatrix}
        \mapsto
        \begin{pmatrix}
            x_1 + g_1 (x_1 - x_2, x_3 - x_4) \\
            x_2 + g_1 (x_1 - x_2, x_3 - x_4) \\
            x_3 + g_2 (x_3 - x_4) \\
            x_4 + g_2 (x_3 - x_4)
        \end{pmatrix}
    \end{equation}

    Analog to the classical two branch Lai--Massey we can describe the generalized Lai--Massey as composition of several GTDS and linear permutations.
    \begin{thm}\label{Th: Lai--Massey reduced to GTDS}
        Let $\Fq$ be a finite field.
        The generalized  Lai--Massey can be constructed via compositions of generalized triangular dynamical systems and affine permutations.
    \end{thm}
    \begin{proof}
        The first dynamical system is the application of the univariate permutation polynomials to the first $m$ branches
        \begin{equation*}
            \mathcal{F}_1 :
            (x_1, \dots, x_n)^\intercal
            \mapsto
            \begin{pmatrix}
                \{ p_i (x_i) \}_{1 \leq i \leq m} \\ \{ x_i \}_{m + 1 \leq i \leq n}
            \end{pmatrix}
            .
        \end{equation*}
        In the second one we construct the sum with the $\omega_i$'s
        \begin{equation*}
            \mathcal{F}_2 :
            \begin{pmatrix}
                x_1 \\ \vdots \\ x_n
            \end{pmatrix}
            \mapsto
            \begin{pmatrix}
                \begin{dcases}
                    \begin{rcases}
                        \omega_i \cdot x_i, & \omega_i \neq 0, \\
                        x_i, & \omega_i = 0
                    \end{rcases}
                \end{dcases}
                _{1 \leq i \leq m - 1} \\
                \sum_{i = 1}^{m} \omega_i \cdot x_i \\
                \{ x_i \}_{m + 1 \leq i \leq n}
            \end{pmatrix}.
        \end{equation*}
        In the third one we add the polynomial $g$ to the first $m - 1$ branches, though we have to do a case distinction whether $\omega_i \neq 0$ or not,
        \begin{equation*}
            \mathcal{F}_3:
            \begin{pmatrix}
                x_1 \\ \vdots \\ x_n
            \end{pmatrix}
            \mapsto
            \begin{pmatrix}
                \begin{dcases}
                    \begin{rcases}
                        x_i + \omega_i \cdot g (x_m, x_{m + 1}, \dots, x_n), & \omega_i \neq 0,\\
                        x_i + g (x_m, x_{m + 1}, \dots, x_n), & \omega_i = 0
                    \end{rcases}
                \end{dcases}
                _{1 \leq i \leq m - 1} \\
                \{ x_i \}_{m \leq i \leq n}
            \end{pmatrix}
        \end{equation*}
        Then we add the polynomial $g$ to the $m$\textsuperscript{th} branch and cancel the factors $\omega_i$ whenever necessary
        \begin{equation*}
            \mathcal{F}_4:
            \begin{pmatrix}
                x_1 \\ \vdots \\ x_n
            \end{pmatrix}
            \mapsto
            \begin{pmatrix}
                \begin{dcases}
                    \begin{rcases}
                        \omega_i^{-1} \cdot x_i, & \omega_i \neq 0, \\
                        x_i, & \omega_i = 0
                    \end{rcases}
                \end{dcases}
                _{1 \leq i \leq m - 1} \\
                \omega_m^{-1} \cdot \left( x_m - \sum_{ \substack{1 \leq i \leq m - 1 \\ \omega_i \neq 0 }} x_i \right) \\
                \{ x_i \}_{m + 1 \leq i \leq n}
            \end{pmatrix}
            .
        \end{equation*}
        Lastly, we apply the univariate permutation polynomials to the remaining branches
        \begin{equation*}
            \mathcal{F}_5:
            (x_1, \dots, x_n)^\intercal
            \mapsto
            \begin{pmatrix}
                \{ x_i \}_{1 \leq i \leq m} \\
                \{ p_i (x_i) \}_{m + 1 \leq i \leq n}
            \end{pmatrix}
            .
        \end{equation*}
        Now it follows from a simple calculation that indeed $\mathcal{F}_5 \circ \dots \circ \mathcal{F}_1$ implements the generalized Lai--Massey construction. \qed
    \end{proof}

    \subsection{Constructions with Non-Trivial Polynomials with No Zeros}
    Recall that for $1 \leq i \leq n - 1$ the $i$\textsuperscript{th} branch in a GTDS is given by
    \begin{equation}
        f_i (x_1, \dots, x_n) = p_i (x_i) \cdot g_i (x_{i + 1}, \dots, x_n) + h_i (x_{i + 1}, \dots, x_n),
    \end{equation}
    where $g_i$ is a polynomial that does not have any zeros.
    All constructions we have investigated so far have one thing in common, they all use trivial $g_i$'s, that is $g_i = 1$.
    Therefore, it is now time to cover constructions that have non-trivial $g_i$'s.

    \subsubsection{Horst \& Griffin.}
    The \texttt{Horst} scheme \cite{Griffin} was introduced as generalization of the Feistel scheme.
    It is defined as
    \begin{equation}
        \begin{pmatrix}
            x_1 \\ \vdots \\ x_n
        \end{pmatrix}
        \mapsto
        \begin{pmatrix}
            x_1 \cdot g_1 (x_2, \dots, x_n) + h_1 (x_2, \dots, x_n) \\
            \vdots \\
            x_{n - 1} \cdot g_{n - 1} (x_n) + h_{n - 1} (x_n) \\
            x_n
        \end{pmatrix}
        ,
    \end{equation}
    where $g_i, h_i \in \Fq [x_{i + 1}, \dots, x_n]$.
    If the polynomials $g_i$'s do not have any zeros over $\Fq$, then \texttt{Horst} induces a permutation.
    Clearly, this is a special instance of a GTDS.
    The permutation \Griffin-\textpi \cite{Griffin} is a concatenation of a SPN and a \texttt{Horst} permutation, so it is also covered by the GTDS framework.

    \subsubsection{\ReinforcedConcrete.}
    The \ReinforcedConcrete \cite{ReinforcedConcrete} hash function is the first arithmetization-oriented hash function that utilizes lookup tables.
    At round level the \texttt{Reinforced Concrete} permutation over $\Fp^3$, where $p \gtrsim 2^{64}$ is a prime, consists of three small permutations.
    The first permutation is the mapping \texttt{Bricks}
    \begin{equation}
        \begin{split}
            \texttt{Bricks}: \Fp^3 &\to \Fp^3, \\
            \begin{pmatrix}
                x_1 \\ x_2 \\ x_3
            \end{pmatrix}
            &\mapsto
            \begin{pmatrix}
                x_1^d \\
                x_2 \cdot \left( x_1^2 + \alpha_1 \cdot x_1 + \beta_1 \right) \\
                x_3 \cdot \left( x_2^2 + \alpha_2 \cdot x_2 + \beta_2 \right)
            \end{pmatrix},
        \end{split}
    \end{equation}
    where $d = 5$, note that the prime must be suitable chosen such that $\gcd \left( d, p - 1 \right) \allowbreak = 1$ else the first component does not induce a permutation, and $\alpha_1, \alpha_2, \beta_1, \beta_2 \in \Fp$ such that $\alpha_i^2 - 4 \beta_i$ is not quadratic residue module $p$, then the quadratic polynomials do not have any zeros over $\Fp$.
    The second permutation is called \texttt{Concrete} and is given by matrix multiplication and constant addition.
    The third permutation \texttt{Bars} is an S-box that is implemented via a lookup table.
    Clearly, these mappings are covered by the GTDS framework.

    \subsubsection{\Arion.}
    The \Arion block cipher and \ArionHash hash function \cite{Arion} are the first designs that utilize the full GTDS structure at round level.
    It is defined over prime fields with $p \geq 2^{60}$, and its GTDS is
    \begin{equation}
        \begin{split}
            f_i (x_1, \dots, x_n) &= x_i^{d_1} \cdot g_i (\sigma_{i + 1, n}) + h_i (\sigma_{i + 1, n}), \qquad 1 \leq i \leq n - 1, \\
            f_n (x_1, \dots, x_n) &= x_n^{e},
        \end{split}
    \end{equation}
    where $d_1 \in \mathbb{Z}_{> 1}$ is the smallest integer such that $\gcd \left( d_1, p - 1 \right) = 1$, for one $d_2 \in \{ 121, 123, 125, 129, 161, 257 \}$ $e \in \mathbb{Z}_{ >1}$ is such that $e \cdot d_2 \equiv 1 \mod p - 1$, $g_i, h_i \in \Fp [x]$ are quadratic polynomials such that the $g_i$'s are irreducible, and
    \begin{equation}
        \sigma_{i + 1, n} = \sum_{j = i + 1}^{n} x_i + f_i.
    \end{equation}

    \section{Analysis of GTDS-based Permutations}\label{Sec: analysis of GTDS}
    \subsection{Bounding the Differential Uniformity of the GTDS}\label{Sec: Differential uniformity}
    Differential cryptanalysis \cite{C:BihSha90} and its variants are one of the most widely used attack vectors in modern cryptography.
    It is based on the observation that certain input differences can propagate through the rounds of a block cipher with high probability.
    The key measure to quantify whether a function is weak to differential cryptanalysis is the so-called differential uniformity.
    In this section we prove an upper bound for the differential uniformity of the GTDS under minimal assumptions on the polynomials $p_i$, $g_i$ and $h_i$.
    We recall the definition of differential uniformity.
    \begin{defn}[{see \cite{EC:Nyberg93}}]
        Let $\Fq$ be a finite field, and let $f: \Fqn \to \Fqm$ be a function.
        \begin{enumerate}
            \item The differential distribution table of $f$ at $\mathbf{a} \in \Fqn$ and $\mathbf{b} \in \Fqm$ is defined as
            \begin{equation*}
                \delta_f (\mathbf{a}, \mathbf{b}) = \left\vert \{ \mathbf{x} \in \Fqn \mid f (\mathbf{x} + \mathbf{a}) - f (\mathbf{x}) = \mathbf{b} \} \right\vert.
            \end{equation*}

            \item The differential uniformity of $f$ is defined as
            \begin{equation*}
                \delta (f) = \max_{ \substack{\mathbf{a} \in \Fqn \setminus \{ \mathbf{0} \},\\ \mathbf{b} \in \Fqm} } \delta_f (\mathbf{a}, \mathbf{b}).
            \end{equation*}
        \end{enumerate}
    \end{defn}

    The following lemma is certainly well-known, it will play an essential role in the proof of the main result of this section.
    \begin{lem}\label{Lem: degree and differential uniformity}
        Let $\Fq$ be a finite field, and let $f \in \Fq [x] / (x^q - x)$.
        Then $\delta (f) < q$ if and only if $\deg \big( f(x + a) - f(x) \big) > 0$ for all $a \in \Fqx$.
        In particular, if $\delta (f) < q$ then $\delta (f) < \degree{f}$.
    \end{lem}
    \begin{proof}
        ``$\Leftarrow$'': By assumption, for all $a \in \Fqx$ and all $b \in \Fq$ we have that $f(x + a) - f(x) - b$ is a non-constant polynomial whose degree is less than $\degree{f}$, so we have that $\delta (f) < \degree{f} < q$.

        ``$\Rightarrow$'': Suppose there exists an $a \in \Fqx$ such that $\deg \big( f (x - a) - f (x) \big) \leq 0$.\footnote{Some textbooks define $\degree{0} = -1$ or $\degree{0} = -\infty$, hence the inequality.}
        Then we can find $b \in \Fq$ such that $f(x + a) - f(x) - b = 0$, so $\delta (f) = q$.
        Now the claim follows by contraposition. \qed
    \end{proof}

    Let us now compute an upper bound for the differential uniformity of a GTDS.
    \begin{thm}\label{Th: differential distribution of GTDS}
        Let $\Fq$ be a finite field, let $n \geq 1$ be an integer, and let $\mathcal{F}: \Fqn \to \Fqn$ be a GTDS.
        Let $p_1, \dots, p_n \in \Fq [x] / (x^q - x)$ be the univariate permutation polynomials of the GTDS $\mathcal{F}$ such that for every $i$ either
        \begin{enumerate}[label=(\roman*)]
            \item $\degree{p_i} = 1$, or

            \item $\degree{p_i} \geq 2$ and $\delta (p_i) < q$.
        \end{enumerate}
        Let $\boldsymbol{\Delta x}, \boldsymbol{\Delta y} \in \Fqn$ be such that $\boldsymbol{\Delta x} \neq \mathbf{0}$.
        Then the differential distribution table of $\mathcal{F}$ at $\boldsymbol{\Delta x}$ and $\boldsymbol{\Delta y}$ is bounded by
        \begin{equation*}
            \begin{split}
                \delta_\mathcal{F} ( \boldsymbol{\Delta x}, \boldsymbol{\Delta y} )
                \leq
                &
                \begin{dcases}
                    \begin{rcases}
                        \delta (p_n), & \boldsymbol{\Delta x}_n \neq 0, \\
                        q, & \boldsymbol{\Delta x}_n, \boldsymbol{\Delta y}_n = 0, \\
                        0, & \boldsymbol{\Delta x}_n = 0,\ \boldsymbol{\Delta y}_n \neq 0
                    \end{rcases}
                \end{dcases}
                \\
                &\phantom{X}\cdot \prod_{i = 1}^{n - 1}
                \begin{dcases}
                    \begin{rcases}
                        \degree{p_i}, & \boldsymbol{\Delta x}_i \neq 0,\ \degree{p_i} > 1, \\
                        q, & \boldsymbol{\Delta x}_i \neq 0,\ \degree{p_i} = 1, \\
                        q, & \boldsymbol{\Delta x}_i = 0
                    \end{rcases}
                \end{dcases}
                .
            \end{split}
        \end{equation*}
    \end{thm}
    \begin{proof}
        Suppose we are given the differential equation
        \begin{equation}\label{Equ: differential equation 1}
            \mathcal{F} (\mathbf{x} + \boldsymbol{\Delta x}) - \mathcal{F} (\mathbf{x}) = \boldsymbol{\Delta y},
        \end{equation}
        Then, the last component of the differential equation only depends on the variable $x_n$, i.e.,
        \begin{equation*}
            p_n (x_n + \boldsymbol{\Delta x}_n) - p_n (x_n) = \boldsymbol{\Delta y}_n.
        \end{equation*}
        If $\boldsymbol{\Delta x}_n \neq 0$, then this equation has at most $\delta (p_n)$ many solutions.
        If $\boldsymbol{\Delta x}_n = \boldsymbol{\Delta y}_n = 0$, then this equation has $q$ many solutions for $x_n$.
        Lastly, if $\boldsymbol{\Delta x}_n = 0$ and $\boldsymbol{\Delta y}_n \neq 0$, then there cannot be any solution for $x_n$.

        Now suppose we have a solution for the last component, say $\hat{x}_n \in \Fq$.
        Then, we can substitute it in \Cref{Equ: differential equation 1} into the $(n - 1)$\textsuperscript{th} component
        \begin{equation*}
            f_{n - 1} (x_{n - 1} + \boldsymbol{\Delta x}_{n - 1}, \hat{x}_n + \boldsymbol{\Delta x}_n) - f_{n - 1} (x_{n - 1}, \hat{x}_n) = \boldsymbol{\Delta y}_{n - 1}.
        \end{equation*}
        Since $\hat{x}_n$ is a field element we can reduce this equation to
        \begin{equation}\label{Equ: differential equation 2}
            \alpha \cdot p_{n - 1} (x_{n - 1} + \boldsymbol{\Delta x}_{n - 1} ) - \beta \cdot p_{n - 1} (x_{n - 1}) + \gamma = \boldsymbol{\Delta y}_{n - 1},
        \end{equation}
        where $\alpha, \beta, \gamma \in \Fq$ and $\alpha, \beta \neq 0$.
        Now we have to do a case distinction on the various case for $\alpha$, $\beta$, $\boldsymbol{\Delta x}_{n - 1}$ and $\degree{p_{n - 1}}$.
        \begin{itemize}
            \item For $\boldsymbol{\Delta x}_{n - 1} \neq 0$ and $\alpha \neq \beta$, then \Cref{Equ: differential equation 2} has at most $\degree{p_{n - 1}}$ many solutions.

            \item For $\boldsymbol{\Delta x}_{n - 1} \neq 0$, $\alpha = \beta$ and $\degree{p_{n - 1}} > 1$, \Cref{Equ: differential equation 2} is the differential equation for $p_{n - 1}$ scaled by $\alpha$ and by assumption this equation has at most $\delta (p_{n - 1}) < q$ many solutions.
            So we can apply \Cref{Lem: degree and differential uniformity} to immediately conclude that $\delta (p_{n - 1}) < \degree{p_{n - 1}}$.

            \item For $\alpha = \beta$ and $\degree{p_{n - 1}} = 1$, then only constant terms remain in \Cref{Equ: differential equation 2}.
            In principle, it can happen that $\alpha \cdot a_{n - 1, 1} \cdot \boldsymbol{\Delta x}_{n - 1} + \gamma = \boldsymbol{\Delta y}_{n - 1}$, where $a_{n - 1, 1} \in \Fqx$ is the coefficient of the linear term of $p_{n - 1}$.
            So this case can have at most $q$ many solutions.

            \item For $\boldsymbol{\Delta x}_{n - 1} = 0$, then in principle it can happen that $\alpha = \beta$ and $\boldsymbol{\Delta y}_{n - 1} = \gamma$.
            So this case can have at most $q$ many solutions.
        \end{itemize}
        Summarizing these cases we conclude that
        \begin{itemize}
            \item If $\boldsymbol{\Delta x}_{n - 1} \neq 0$ and $\degree{p_{n - 1}} > 1$, then \Cref{Equ: differential equation 2} has at most $\degree{p_{n - 1}}$ many solutions.

            \item If $\boldsymbol{\Delta x}_{n - 1} \neq 0$ and $\degree{p_{n - 1}} = 1$, then \Cref{Equ: differential equation 2} has at most $q$ many solutions.

            \item If $\boldsymbol{\Delta x}_{n - 1} = 0$, then \Cref{Equ: differential equation 2} has at most $q$ many solutions.
        \end{itemize}
        Inductively, we now work upwards through the branches to derive the claim. \qed
    \end{proof}

    Let the function $\wt: \Fqn \to \mathbb{Z}$ denote the Hamming weight, i.e.\ it counts the number of non-zero entries of a vector in $\Fqn$.
    \begin{cor}\label{Cor: differential uniformity}
        Let $\Fq$ be a finite field, let $n \geq 1$ be an integer, and let $\mathcal{F}: \Fqn \to \Fqn$ be a GTDS.
        Let $p_1, \dots, p_n \in \Fq [x] / (x^q - x)$ be the univariate permutation polynomials of the GTDS $\mathcal{F}$, and let $\boldsymbol{\Delta x}, \boldsymbol{\Delta y} \in \Fqn$ be such that $\boldsymbol{\Delta x} \neq \mathbf{0}$.
        If for all $1 \leq i \leq n$ one has that $1 < \degree{p_i} \leq d$ and $\delta (p_i) < q$, then
        \begin{equation*}
            \delta_\mathcal{F} (\boldsymbol{\Delta x}, \boldsymbol{\Delta y}) \leq q^{n - \wt (\boldsymbol{\Delta x})} \cdot d^{\wt (\boldsymbol{\Delta x})}.
        \end{equation*}
        In particular,
        \begin{equation*}
            \prob \left[ \mathcal{F} \! : \boldsymbol{\Delta x} \to \boldsymbol{\Delta y} \right] \leq \left( \frac{d}{q} \right)^{\wt (\boldsymbol{\Delta x})}.
        \end{equation*}
    \end{cor}
    \begin{proof}
        If $f \in \Fq [x]$ is a polynomial such that $f (x + a) - f (x)$ is a non-constant polynomial for all $a \in \Fqx$, then $\delta (f) < \degree{f}$.
        Now we apply this observation to \Cref{Th: differential distribution of GTDS}.
        The bound for the probability from the first and division by $q^n$. \qed
    \end{proof}

    Let $p_1, \dots, p_n \in \Fq [x]$ be univariate permutation polynomials that satisfy the assumption from \Cref{Th: differential distribution of GTDS} and assume that $1 < \delta (p_i) \leq d$ for all $i$.
    Let us consider the SPN
    \begin{equation}\label{Equ: SPN}
        S: \left( x_1, \dots, x_n \right) \mapsto \big( p_1 (x_1), \dots, p_n (x_n) \big).
    \end{equation}
    It is well-known that
    \begin{equation}\label{Equ: bound for SPN}
        \prob \left[ S \! : \boldsymbol{\Delta x} \to \boldsymbol{\Delta y} \right] \leq \left( \frac{d}{q} \right)^{\wt \left( \boldsymbol{\Delta x} \right)}.
    \end{equation}
    Now let $\mathcal{F}: \Fqn \to \Fqn$ be a GTDS with the univariate permutation polynomials $p_1, \dots, p_n$.
    Provided that $\delta (p_i) \approx \degree{p_i}$ when compared to $q$, then we expect that the bound from \Cref{Cor: differential uniformity} almost coincides with \Cref{Equ: bound for SPN}.
    I.e., the GTDS $\mathcal{F}$ and the SPN $S$ are in almost the same security class with respect to differential cryptanalysis.
    What is the contribution of the polynomials $g_i$ and $h_i$ in the GTDS $\mathcal{F}$ then?
    Conceptually, they can only lower the probability compared to the ``SPN bounds'' from \Cref{Equ: bound for SPN} but never increase it.

    Of course, this now raises the question of how this contribution can be incorporated into an improved bound.
    If we recall the proof of the theorem, then we can translate this question into the following problem: Let $f \in \Fq [x]$ be a polynomial, let $\alpha, \beta, \Delta x \in \Fqx$ and $\delta, \Delta y \in \Fq$.
    How many solutions does the equation
    \begin{equation}
        \alpha \cdot f(x + \Delta x) - \beta \cdot f (x) + \gamma = \Delta y
    \end{equation}
    have?
    Moreover, one could try to estimate the codomains of the $g_i$'s and $h_i$'s to exclude values for $\alpha, \beta, \gamma$ than can never arise in the differential equation of the GTDS.

    For the application of \Cref{Th: differential distribution of GTDS} it is crucial that one knows that the univariate permutation polynomials have non-trivial differential uniformity.
    Therefore, we derive two efficient criteria that bypass the computation of the full differential distribution table.
    \begin{lem}\label{Lem: difference polynomial non-constant}
        Let $\Fq$ be a finite field of characteristic $p$, let $a \in \Fqx$, and let $f = \sum_{i = 0}^{d} b_i \cdot x^i \in \Fq [x] / (x^q - x)$ be such that $d = \degree{f} > 1$.
        \begin{enumerate}
            \item\label{Item: criterium 1} If $q$ is prime, then $f (x + a) - f (x)$ is a non-constant polynomial.

            \item\label{Item: criterium 2} If $q$ is a prime power, let $d' = \max \left\{ \degree{f - b_d \cdot x^d}, 1 \right\}$.
            If there exists $d' \leq k \leq d - 1$ such that $\gcd \left( p, \binom{d}{k} \right) = 1$, then $f (x + a) - f (x)$ is a non-constant polynomial.
        \end{enumerate}
    \end{lem}
    \begin{proof}
        For (1), we expand $f$ via the binomial formula
        \begin{align*}
            f (x + a) - f (x)
            &= \sum_{i = 0}^{\degree{f}} b_i \cdot \left( (x + a)^i - x^i \right) \\
            &= \sum_{i = 0}^{\degree{f}} b_i \cdot \left( \sum_{k = 0}^{i - 1} \binom{i}{k} \cdot a^{i - k} \cdot x^k \right) \\
            &= a_d \cdot \binom{d}{d - 1} \cdot a \cdot x^{d - 1} + g (x),
        \end{align*}
        where $\degree{g} < d - 1$.
        Since $d < q$ and $q$ is prime we always have that $\binom{d}{d - 1} \not\equiv 0 \mod q$.

        For (2), the assumption on the binomial coefficient guarantees that at least one binomial coefficient $\binom{d}{k}$, where $d' \leq k \leq d - 1$, is non-zero in $\Fq$. \qed
    \end{proof}

    By \ref{Item: criterium 1}, over prime fields we can apply \Cref{Th: differential distribution of GTDS} for every univariate permutation polynomial of degree greater than $1$.
    With \ref{Item: criterium 2} we can settle some polynomials $f \in \Fq [x] / (x^q - x)$ such that $\gcd \big( q, \deg (f) \big) \neq 1$.
    E.g., let $q = 2^n$, and let $f = x^{2^n - 2}$, then
    \begin{equation}
        \binom{2^n - 2}{2^n - 4} = \left( 2^n - 3 \right) \cdot \left( 2^{n - 1} - 1 \right) \equiv 1 \mod 2.
    \end{equation}

    Finally, let us discuss when \Cref{Th: differential distribution of GTDS} provides viable bounds for differential cryptanalysis.
    Classical symmetric cryptography is designed to be efficiently on bit based hard- and software.
    So these designs can be modeled over $\F_{2^n}^{m}$, where $m, n \geq 1$.
    Though, the polynomial degree of components in these primitives is usually of minor concern in design as well as cryptanalysis.
    As consequence, many designs were proposed that have high polynomial degrees but still can be efficiently evaluated.
    The prime example is the AES S-box which is based on the inversion permutation $x^{q - 2}$.
    If we instantiate a GTDS with the inversion permutation and apply \Cref{Cor: differential uniformity}, then we obtain the bound $\frac{q - 2}{q}$ for the respective component.
    Needless to say that this bound will be hardly of use for cryptanalysis.
    On the other hand, if we take a look to symmetric primitives targeting Multi-Party Computation and Zero-Knowledge protocols, then \Cref{Th: differential distribution of GTDS} becomes viable.
    Typically, these protocols are instantiated over prime fields $p \geq 2^{64}$, and they require a symmetric cipher or hash function which requires a very low number of multiplications for evaluation.
    Moreover, for an univariate permutation polynomial $f \in \Fp [x] / (x^p - x)$ in an AOC designs one often has that $\degree{f} < 2^9$ or $\degree{f^{-1}} < 2^9$, so we obtain a bound which is less than $\frac{2^9}{2^{64}}$ for the respective component.
    For an iterated design this bound is small enough to provide resistance against differential cryptanalysis and its variants.

    We also want to highlight that \Cref{Th: differential distribution of GTDS} has been applied in the differential cryptanalysis of \Arion \cite[\S 3.1]{Arion}.

    \subsection{A Bound on the Correlation of the GTDS}
    Linear cryptanalysis was introduced in \cite{EC:Matsui93} and extended to arbitrary finite fields in \cite{SAC:BaiSteVau07}.
    For the attack one tries to find affine approximations of the rounds of a block cipher for a sample of known plaintexts.
    The key measure to quantify whether a function is weak to linear cryptanalysis is the so-called correlation.
    In this section we will prove an upper bound for the differential uniformity of the GTDS under minimal assumptions on the polynomials $p_i$, $g_i$ and $h_i$.
    We recall the definition of correlation.
    \begin{defn}[{see \cite[Definition~6, 15]{SAC:BaiSteVau07}}]\label{Def: correlation}
        Let $\Fq$ be a finite field, let $n \geq 1$, let $\chi: \Fq \to \mathbb{C}$ be a non-trivial additive character, let $F: \Fqn \to \Fqn$ be a function, and let $\mathbf{a}, \mathbf{b} \in \Fqn$.
        \begin{enumerate}
            \item\label{Item: correlation} The correlation for the character $\chi$ of the linear approximation $(\mathbf{a}, \mathbf{b})$ of $F$ is defined as
            \begin{equation*}
                \CORR_F (\chi, \mathbf{a}, \mathbf{b}) = \frac{1}{q^n} \cdot \sum_{\mathbf{x} \in \Fqn} \chi \Big( \big< \mathbf{a}, F (\mathbf{x}) \big> + \braket{\mathbf{b}, \mathbf{x}} \! \Big).
            \end{equation*}

            \item The linear probability for the character $\chi$ of the linear approximation $(\mathbf{a}, \mathbf{b})$ of $F$ is defined as
            \begin{equation*}
                \LP_F (\chi, \mathbf{a}, \mathbf{b}) = \left| \CORR_F (\chi, \mathbf{a}, \mathbf{b})  \right|^2.
            \end{equation*}
        \end{enumerate}
    \end{defn}
    \begin{rem}
        To be precise Baign\`eres et al.\ \cite{SAC:BaiSteVau07} defined linear cryptanalysis over arbitrary abelian groups, in particular for maximal generality they defined the correlation with respect to two additive characters $\chi, \psi: \Fq \to \mathbb{C}$ as
        \begin{equation}\label{Equ: original correlation definition}
            \CORR_F (\chi, \psi, \mathbf{a}, \mathbf{b}) = \frac{1}{q^n} \cdot \sum_{\mathbf{x} \in \Fqn} \chi \Big( \big< \mathbf{a}, F (\mathbf{x}) \big> \Big) \cdot \psi \Big( \big< \mathbf{b}, \mathbf{x} \big> \! \Big).
        \end{equation}
        Let $\Fq$ be a finite field of characteristic $p$, and let $\Tr: \Fq \to \Fp$ be the absolute trace function, see \cite[2.22.~Definition]{Niederreiter-FiniteFields}.
        For all $x \in \Fq$ we define the function $\chi_1$ as
        \begin{equation*}
            \chi_1 (x) = \exp \left(\frac{2 \pi i}{p} \cdot \Tr (x) \right).
        \end{equation*}
        Then for every non-trivial additive character $\chi: \Fq \to \mathbb{C}$ there exist $a \in \Fqx$ such that $\chi (x) = \chi_1 (a \cdot x)$, see \cite[5.7.~Theorem]{Niederreiter-FiniteFields}.
        Therefore, after an appropriate rescaling that we either absorb into $\mathbf{a}$ or $\mathbf{b}$ we can transform \Cref{Equ: original correlation definition} into \Cref{Def: correlation} \ref{Item: correlation}.
    \end{rem}

    If we linearly approximate every round of a block cipher $\mathcal{C}_r: \Fqn \times \Fq^{n \times (r + 1)} \to \Fqn$, then a tuple $\Omega = (\boldsymbol{\omega}_0, \dots, \boldsymbol{\omega}_r) \subset \left( \Fqn \right)^{r + 1}$ is called a linear trail for $\mathcal{C}_r$, where $(\boldsymbol{\omega}_{i - 1}, \boldsymbol{\omega}_i)$ is the linear approximation of the $i$\textsuperscript{th} round $\mathcal{R}^{(i)}$ of $\mathcal{C}_r$.
    For an additive character $\chi: \Fq \to \mathbb{C}$ and under the assumption that the rounds of $\mathcal{C}_r$ are statistically independent, we denote the linear probability of the linear trail $\Omega$ of $\mathcal{C}_r$ by
    \begin{equation}
        \LP_{\mathcal{C}_r} (\chi, \Omega) = \prod_{i = 1}^{r} \LP_{\mathcal{R}^{(i)}} (\chi, \boldsymbol{\omega}_{i - 1}, \boldsymbol{\omega}_i).
    \end{equation}
    If a distinguisher is limited to $N$ queries, then by \cite[Theorem~7]{SAC:BaiSteVau07} the advantage of a linear distinguisher with a single linear trail is lower bounded under heuristic assumptions by
    \begin{equation}
        p_{success} \succeq 1 - e^{-\frac{N}{4} \cdot \LP_{\mathcal{C}_r} (\chi, \Omega)}.
    \end{equation}
    Moreover, for any function $F: \Fqn \to \Fqn$ and any $\mathbf{A} \in \Fqnxn$ and $\mathbf{c} \in \Fqn$ one has
    \begin{equation}
        \LP_{\mathbf{A} F + \mathbf{c}} (\chi, \mathbf{a}, \mathbf{b}) = \LP_F (\chi, \mathbf{A}^\intercal \mathbf{a}, \mathbf{b}).
    \end{equation}
    Therefore, bounding the correlation of the GTDS is the key ingredient to estimate the resistance of a block cipher against linear cryptanalysis.

    As preparation, we prove a bound on univariate character sums which follows as corollary to \cite[5.38. Theorem]{Niederreiter-FiniteFields}.
    \begin{lem}\label{Lem: exponential sum}
        Let $\Fq$ be a finite field, let $\chi: \Fq \to \mathbb{C}$ be a non-trivial additive character, let $f \in \Fq [x]$ be a permutation polynomial such that $\gcd \left( \degree{f}, q \right) = \gcd \left( \degree{f^{-1}} , q \right) = 1$, and let $a, b \in \Fqx$.
        Then
        \begin{equation*}
            \left| \sum_{x \in \Fq} \chi \left( a \cdot f (x) + b \cdot x\right) \right| \leq \Big( \min \Big\{ \degree{f}, \degree{f^{-1}} \Big\} - 1 \Big) \cdot q^{1 / 2}.
        \end{equation*}
    \end{lem}
    \begin{proof}
        Since $f$ is a permutation polynomial we can rewrite the character sum
        \begin{align*}
            \sum_{x \in \Fq} \chi \left( a \cdot f (x) + b \cdot x \right)
            &= \sum_{y \in \Fq} \chi \Big( a \cdot f \left( f^{-1} (y) \right) + b \cdot f^{-1} (y) \Big) \\
            &= \sum_{y \in \Fq} \chi \Big( a \cdot y + b \cdot f^{-1} (y) \Big),
        \end{align*}
        where the second equality follows from $f \big( f^{-1} (x) \big) \equiv x \mod \left( x^q - x \right)$.
        By our assumptions we can then apply the Weil bound \cite[5.38. Theorem]{Niederreiter-FiniteFields} to obtain the inequality. \qed
    \end{proof}

    Now we can compute an upper bound on the correlation of the GTDS.
    \begin{thm}\label{Th: bound on correlation}
        Let $\Fq$ be a finite field, let $n \geq 1$, let $\chi: \Fq \to \mathbb{C}$ be a non-trivial additive character, let $\mathcal{F} = \{ f_1, \dots, f_n \} \subset \Fq [x_1, \dots, x_n]$ be a GTDS, let $p_1, \dots, p_n \in \Fq [x]/ (x^q - x)$ be the univariate permutation polynomials in the GTDS $\mathcal{F}$ such that $\gcd \left( \degree{p_i}, q \right) = \gcd \left( \degree{p_i^{-1}}, q \right) = 1$ for all $1 \leq i \leq n$, and let $\mathbf{a}, \mathbf{b} \in \Fqn$.
        If $\mathbf{a} \neq \mathbf{0}$ denote with $1 \leq j \leq n$ the first index such that $a_j \neq 0$.
        Then
        \begin{equation*}
            \left| \CORR_\mathcal{F} (\chi, \mathbf{a}, \mathbf{b}) \right| \leq
            \begin{dcases}
                1, & \mathbf{a}, \mathbf{b} = \mathbf{0}, \\
                0, &
                \left\{
                \begin{array}{l}
                    \mathbf{a} = \mathbf{0},\ \mathbf{b} \neq \mathbf{0}, \\
                    \mathbf{a} \neq \mathbf{0},\ \mathbf{b} = \mathbf{0},
                \end{array}
                \right. \\
                0, & b_j = 0, \\
                1, & b_j \neq 0,\ \degree{p_j} = 1, \\
                \frac{\min \Big\{ \degree{p_j}, \degree{p_j^{-1}} \Big\} - 1}{\sqrt{q}}, & b_j \neq 0,\ \degree{p_j} > 1.
            \end{dcases}
        \end{equation*}
    \end{thm}
    \begin{proof}
        The first case is trivial, for the second and the third we recall that any non-trivial linear combination of an orthogonal system is a multivariate permutation polynomial, cf.\ \cite[7.39.~Corollary]{Niederreiter-FiniteFields}.
        Recall that for any multivariate permutation polynomial $f \in \Fq [x_1, \dots, x_n]$ the equation $f (x_1, \dots, x_n) = \alpha$ has $q^{n - 1}$ many solutions for every $\alpha \in \Fq$.
        So the exponential sum of the correlation collapses to
        \begin{equation*}
            q^{n - 1} \cdot \sum_{x \in \Fq} \chi (x) = 0,
        \end{equation*}
        which is zero by \cite[5.4. Theorem]{Niederreiter-FiniteFields}.

        Now let us assume that $a_j \neq 0$.
        Then we apply the triangular inequality to the variables $x_{j + 1}, \dots, x_n$ as follows
        \begin{align*}
            &\left| \sum_{\mathbf{x} \in \Fqn} \chi \Big( \big< \mathbf{a}, F(\mathbf{x}) \big> + \braket{\mathbf{b}, \mathbf{x}} \! \Big) \right| \\
            &= \left| \sum_{\mathbf{x} \in \Fqn} \chi \left( \sum_{i = j + 1}^{n} a_i \cdot f_i (\mathbf{x}) + b_i \cdot x_i \right) \cdot \chi \Big( a_j \cdot f_j (\mathbf{x}) + b_j \cdot x_j \Big) \right| \\
            & \leq \sum_{x_{j + 1}, \dots, x_n \in \Fq} \left| \chi \left( \sum_{i = j + 1}^{n} a_i \cdot f_i (\mathbf{x}) + b_i \cdot x_i \right) \cdot \sum_{x_1, \dots, x_j \in \Fq} \chi \Big( a_j \cdot f_j  (\mathbf{x}) + b_j \cdot x_j \Big) \right| \\
            &= \sum_{x_{j + 1}, \dots, x_n \in \Fq} \left| \chi \left( \sum_{i = j + 1}^{n} a_i \cdot f_i (\mathbf{x}) + b_i \cdot x_i \right) \right| \cdot \left| \sum_{x_1, \dots, x_j \in \Fq} \chi \Big( a_j \cdot f_j  (\mathbf{x}) + b_j \cdot x_j \Big) \right| \\
            &= \sum_{x_{j + 1}, \dots, x_n \in \Fq} \left| q^{j - 1} \cdot \sum_{x_j \in \Fq} \chi \Big( a_j \cdot f_j (x_j, \dots, x_n) + b_j \cdot x_j \Big) \right|
            = (\ast).
        \end{align*}
        For any fixed $(x_{j + 1}, \dots, x_n) \in \Fq^{n - j}$ we have that
        \begin{equation*}
            \hat{f}_j (x_j) = a_j \cdot f_j (x_j, \dots, x_n) + b_j \cdot x_j = a_j \cdot \big( p_j (x_j) \cdot \alpha + \beta \big) + b_j \cdot x_j,
        \end{equation*}
        where
        \begin{align*}
            \alpha &= g_j (x_{j + 1}, \dots, x_n) \in \Fqx, \\
            \beta &= h_j (x_{j + 1}, \dots, x_n) \in \Fq.
        \end{align*}
        If $b_j = 0$, then $\hat{f}_j$ is a univariate permutation polynomial in $x_j$.
        So the exponential sum inside the absolute value of $(\ast)$ must vanish for every $(x_{j + 1}, \dots, x_n) \in \Fq^{n - j}$.

        For $b_j \neq 0$, if $\degree{p_j} = 1$, then in principle $\hat{f}_j$ can be a constant polynomial.
        Since we do not know for how many $(x_{j + 1}, \dots, x_n) \in \Fq^{n - j}$ this happens we have to use the trivial bound.

        For the final case $\degree{p_j} > 1$, recall that we assumed
        \begin{equation*}
            \gcd \left( \degree{p_i}, q \right) = \gcd \left( \degree{p_i^{-1}}, q \right) = 1.
        \end{equation*}
        for all $1 \leq i \leq n$.
        So for every fixed $(x_{j + 1}, \dots, x_n) \in \Fq^{n - j}$ we can now apply \Cref{Lem: exponential sum} to bound the absolute value in $(\ast)$.
        This yields
        \begin{align*}
            (\ast)
            &\leq p^{j - 1} \cdot \sum_{k = j + 1}^{n} \sum_{x_k \in \Fq} \Big( \min \left\{ \degree{p_j}, \degree{p_j^{-1}} \right\} - 1 \Big) \cdot p^{1 / 2} \\
            &= p^{n - 1 / 2} \cdot \Big( \min \left\{ \degree{p_j}, \degree{p_j^{-1}} \right\} - 1 \Big),
        \end{align*}
        which concludes the proof. \qed
    \end{proof}

    Note if $q$ is a prime number and $f \in \Fq [x] / (x^q - x)$, then the coprimality condition is always satisfied.
    \begin{cor}\label{Cor: bound on linear probability}
        In the scenario of \Cref{Th: bound on correlation}, for any non-trivial additive character $\chi: \Fq \to \mathbb{C}$ one has
        \begin{equation*}
            \LP_\mathcal{F} (\chi, \mathbf{a}, \mathbf{b}) \leq
            \begin{dcases}
                1, & \mathbf{a}, \mathbf{b} = \mathbf{0}, \\
                0, &
                \left\{
                \begin{array}{l}
                    \mathbf{a} = \mathbf{0},\ \mathbf{b} \neq \mathbf{0}, \\
                    \mathbf{a} \neq \mathbf{0},\ \mathbf{b} = \mathbf{0},
                \end{array}
                \right. \\
                0, & b_j = 0, \\
                1, & b_j \neq 0,\ \degree{p_j} = 1, \\
                \frac{\bigg( \min \Big\{ \degree{p_j}, \degree{p_j^{-1}} \Big\} - 1 \bigg)^2}{q}, & b_j \neq 0,\ \degree{p_j} > 1.
            \end{dcases}
        \end{equation*}
    \end{cor}

    Analog to the differential uniformity, let us compare \Cref{Cor: bound on linear probability} to the SPN $S$ from \Cref{Equ: SPN} with the additional assumptions that $\gcd \left( \degree{p_i}, q \right) = \gcd \left( \degree{p_i^{-1}}, q \right) = 1$ and $1 < \degree{p_i} \leq d$ or $1 < \degree{p_i^{-1}} \leq d$ for all $1 \leq i \leq n$.
    It is well-known that
    \begin{equation}
        \LP_S \left( \chi, \mathbf{a}, \mathbf{b} \right) \leq
        \begin{dcases}
            0, &
            \begin{array}{l}
                \exists i \! : a_i \neq 0,\ b_i = 0, \\
                \exists i \! : a_i = 0,\ b_i \neq 0,
            \end{array}
            \\
            \left( \frac{(d - 1)^2}{q} \right)^{\wt \left( \mathbf{a} \right)}, & else.
        \end{dcases}
    \end{equation}
    Since this probability decreases with $\mathcal{O} \left( q^{- \wt \left( \mathbf{a} \right)} \right)$ it might in principle be suitable to estimate the linear hull of a SPN cipher.
    On the other hand, our bound from \Cref{Cor: bound on linear probability} if non-trivial is always in $\mathcal{O} \left( q^{-1} \right)$.
    While it is still possible to estimate the probability of linear trails of a GTDS cipher with this bound, it is not suitable to estimate the linear hull of a GTDS cipher.

    Analog to the differential uniformity bound for the GTDS, we do not expect that \Cref{Th: bound on correlation} and \Cref{Cor: bound on linear probability} will be of great use for binary designs with $q = 2^m$.
    First, the polynomial degree is again the main ingredient of the bound which can be close to $q$ for binary designs.
    Second, in characteristic $2$ the coprimality condition restricts us to univariate permutation polynomials of odd degree.
    In particular, \Cref{Th: bound on correlation} cannot be applied to $x^{2^m - 2}$.
    On the other hand, \Cref{Th: bound on correlation} is also tailored for application to arithmetization-oriented designs over prime fields.
    For prime fields the coprimality condition is always satisfied, and the univariate permutation polynomials in these designs have a suitable small polynomial degree compared to $q$.

    In particular, we want to highlight that \Cref{Th: bound on correlation} has been applied in the linear cryptanalysis of \Arion \cite[\S 3.1]{Arion}.

    \section{Discussion}
    \subsection{Algebraic Frameworks Beyond the GTDS}
    It is worth noting that our GTDS framework is not the first attempt to unify block cipher design strategies.
    In \cite{Yun-QuasiFeistel} the quasi-Feistel cipher idea was introduced. It provides a unified framework for Feistel and Lai--Massey ciphers.
    While our approach utilizes the full algebraic structure of finite fields, the quasi-Feistel cipher uses a contrarian approach by requiring as little algebraic structure as possible.
    In particular, they demonstrate that invertible Feistel and Lai--Massey ciphers can be instantiated over quasigroups (cf.\ \cite{Smith-QuasiGroups}).
    Furthermore, this little algebraic structure is already sufficient to prove theoretical security bounds in the Luby-Rackoff model for quasi-Feistel ciphers.

    \subsection{Hash Functions}
    Our analysis and discussion have been focused on iterative permutations. For most instantiations of known hash functions, a (fixed key) permutation is used to build the compression function.
    Then, iterating the compression function a hash function is built over an arbitrary domain. Our generic description of (keyed) permutations may be viewed as a vector of functions over $\F_q^n$.
    Thus, such permutations can be used to define a hash function $H: \F_q^* \rightarrow \F_q^t$ where the domain of $H$ is of arbitrary length over $\F_q$ and the hash value is of length $t > 0 $ over $\F_q$.
    For example, an instantiation of GTDS-based permutations can be used in a sponge mode \cite{Sponge,EC:BDPV08} to define such a hash function.
    Thus, all our analysis can be easily extrapolated to hash functions.

    \subsection{Beyond Permutations}
    The different conditions on the polynomials defining the GTDS are imposed to ensure that the resulting system is invertible. However, these conditions can be dropped if the goal is not to construct a permutation but possibly a pseudo-random function. Potentially, such a GTDS (without the necessary constraints for invertibility) can be used to construct PRFs over $\Fq$ and is an interesting direction for future work.

    \subsubsection*{Acknowledgments.}
    Matthias Steiner was supported by the KWF under project number KWF-3520|31870|45842.

    \bibliographystyle{splncs04.bst}
    \bibliography{abbrev0.bib,crypto.bib,literature.bib}

\end{document}